\documentclass[letterpaper,10pt,conference]{ieeeconf}

\IEEEoverridecommandlockouts
\newif\ifshowcomments
\showcommentstrue % uncomment to show comments
\usepackage{amsmath, amsfonts, amssymb, mathtools, dsfont, circledsteps, bm}

\usepackage{amsthm}
\usepackage{adjustbox}
\usepackage{bbm}
\usepackage{stmaryrd}
\usepackage{comment}
\usepackage{xcolor}
\usepackage{todonotes}
\usepackage{comment}
\usepackage{environ}
\let\labelindent\relax
\usepackage{enumitem}
\setlist[description]{font=\normalfont\itshape}
\usepackage[font=footnotesize]{caption}
\usepackage{subcaption}
\usepackage{cite}
\usepackage{hyperref}
\usepackage{enumitem}
\usepackage{placeins}
\usepackage[capitalise]{cleveref}

\renewcommand{\baselinestretch}{0.97}
\newtheorem{theorem}{Theorem}

\newtheorem{corollary}{Corollary}
\newtheorem{lemma}{Lemma}
\theoremstyle{definition}
\newtheorem{definition}{Definition}
\newtheorem{problem}{Problem}
\theoremstyle{remark}
\newtheorem{remark}{Remark}
\newtheorem{interpret}{Interpretation}
\newtheorem{example}{Example}
\newtheorem{assumption}{Assumption}
\newtheorem{notation}{Notation}

\newsavebox{\fitmathbox}
\NewEnviron{fitdisplaymath}{%
    \begingroup
    \sbox{\fitmathbox}{$\displaystyle \BODY$}%
    \[
    \ifdim\wd\fitmathbox>\linewidth
        \resizebox{\linewidth}{!}{$\displaystyle \BODY$}%
    \else
        \usebox{\fitmathbox}%
    \fi
    \]
    \endgroup
}

\newcommand{\R}{\mathbb R}

\newcommand{\dd}{\,\mathrm d}

\newcommand{\F}{\mathcal F}
\newcommand{\face}{\trianglelefteqslant}

\newcommand{\res}[3]{{#1}_{#2 \face #3}}
\renewcommand{\epsilon}{\varepsilon}
\renewcommand{\phi}{\varphi}

\definecolor{techgold}{HTML}{B3A369}
\definecolor{techblue}{HTML}{003057}
\definecolor{buzzgold}{HTML}{EAAA00}
\definecolor{greymatter}{HTML}{54585A}
\definecolor{uforange}{HTML}{FA4616}
\definecolor{ufblue}{HTML}{0021A5}
\definecolor{shale}{HTML}{54585A}

\usepackage{autonum}

\title{\bf Multi-Agent System Identification with Nonlinear Sheaf Diffusion}
\author{Nivar Anwer, Hans Riess, and Matthew Hale%
\thanks{Anwer is with the Department of Computer Science, Georgia Tech. Riess and Hale are with the Department of Electrical \& Computer Engineering, Georgia Tech.
Riess was supported by DARPA (HR0011-25-3-0235); Riess and Hale
were supported by AFOSR (FA9550-23-1-0120, FA9550-19-1-0169) and
ONR (N00014-22-1-2435).
Emails: \texttt{\{nanwer3,riess,mhale30\}@gatech.edu}.}}
\date{}

\begin{document}

\maketitle
\raggedbottom

\begin{abstract}
    Local interaction laws governing multi-agent systems can be difficult to recover from trajectory data, even when the dynamics are observed faithfully. In systems governed by a nonlinear sheaf Laplacian---a generalization of the graph Laplacian accommodating heterogeneous state spaces and asymmetric communication channels---the coordination law is encoded by edge potential functions whose gradients produce the inter-agent forces. Because trajectory observations record node-state evolution, they expose only the aggregate effect of the edge forces at each node: distinct interaction laws that agree at the node level are indistinguishable from trajectory data alone. We show that the fundamental obstruction to recovery is topological, measured by sheaf cohomology, and that unique recovery from an unconstrained function class is possible if and only if this cohomology vanishes. When the obstruction is nontrivial, we show that recovery within a finite-dimensional parameterized class is possible precisely when a data-dependent information matrix is positive definite. Experiments validate the theory and illustrate that accurate trajectory reproduction need not certify recovery of the underlying interaction law.
\end{abstract}

\section{Introduction}

% In multi-agent systems (MAS) encompassing robotic swarms, opinion networks, and distributed estimation, collective behavior emerges from local interaction laws---rules specifying what neighboring agents compare and how disagreement is converted into a coordinating force. The classical framework, grounded in spectral graph theory~\cite{chungSpectralGraphTheory1997,saberAgreementProblemsNetworks2003,OlfatiSaber2007,Oh2015}, models these laws via a graph Laplacian acting uniformly on a shared state space.
In multi-agent systems (MAS) encompassing robotic swarms, opinion networks, and distributed estimation, collective behavior emerges from \emph{local interaction laws}---rules specifying what neighboring agents compare and how disagreement is converted into a coordinating force. The classical framework, grounded in spectral graph theory~\cite{chungSpectralGraphTheory1997,saberAgreementProblemsNetworks2003,OlfatiSaber2007,Oh2015}, models these laws via a graph Laplacian acting uniformly on a shared state space.
Cellular sheaves~\cite{currySheavesCosheavesApplications2014} and their spectral theory~\cite{hansenSheafLaplacianSpectral2019} extend this classical framework to the heterogeneous setting: a sheaf over a graph assigns distinct local state spaces and shared comparison spaces to vertices and edges (respectively), with restriction maps specifying how agent states project into communication channels. The nonlinear sheaf Laplacian, formed by composing these restriction maps with edge potentials, subsumes consensus, formation control, opinion dynamics, and target tracking within a single operator~\cite{hanksDistributedMultiAgentCoordination2025,zhaoAsynchronousNonlinearSheaf2025,hanksHeterogeneousMultiAgentMultiTarget2025}. Sheaf Laplacians have also attracted growing attention in graph learning and structural causal models~\cite{bodnarNeuralSheafDiffusion2022,battiloroTangentBundleConvolutional2024,zaghenSheafDiffusionGoes2024,hansenGebhartSheafNeuralNetworks2020,dacuntoLearningConsistentCausal2026}---works that learn the sheaf operator for a downstream task, in contrast to our inverse problem of recovering a fixed edge potential from trajectory data.

In this paper, we study the inverse problem: when can the local interaction law be recovered by observing trajectories of a MAS? A growing literature addresses this question for interacting particle systems~\cite{Lu2019,Li2021Identifiability,luMaggioniTangHeterogeneous2021,Miller2023}, establishing identifiability and consistency under coercivity or excitation conditions, yet predominantly assuming \emph{homogeneous} interactions over a shared state space. A related line of work learns topological structure from data: Hansen \& Ghrist learn sheaf Laplacians from smooth signals~\cite{hansenLearningSheafLaplacians2019,dininoLearningStructureConnection2026}, while other approaches recover graph topology from observed dynamics~\cite{segarraNetworkInferenceConsensus2017,zhuConsensusUnknown2020,talukdarPhysicsInformedTopology2020}.
While the first sheaf cohomology is well understood as an obstruction to coordination \cite{hansenOpinionDynamicsDiscourse2021}, its role as an obstruction to recovering a nonlinear edge potential from trajectory data in a dynamical MAS setting has not yet been characterized.

The present paper addresses this gap. Our contributions are as follows.
\begin{enumerate}
    \item We show that the fundamental obstruction to edge-law recovery from trajectory data is topological: the space of edge forces invisible to node-level observations coincides with $\ker\delta^*$, which by the Hodge decomposition is isomorphic to the first sheaf cohomology $H^1(G;\mathcal{F})$. Recovery from an unconstrained function class is possible if and only if $H^1(G;\mathcal{F}) = 0$.
    \item When $H^1(G;\mathcal{F}) \neq 0$, we prove that identifiability is restored within a finite-dimensional parameterized class whenever a data-dependent information matrix is positive definite.
    \item Experiments on formation transfer, bounded-confidence dynamics, and finite-basis models validate the theory across all three failure modes: cohomological ambiguity, weak excitation, and basis rank deficiency.
\end{enumerate}

The remainder of the paper is organized as follows. Section~\ref{sec:preliminaries} introduces Euclidean sheaves, the nonlinear sheaf Laplacian, and the system identification problem. Section~\ref{sec:nonparametric} establishes the fundamental limitations of nonparametric recovery. Section~\ref{sec:parametric} develops the parametric identifiability criterion and least-squares estimator. Section~\ref{sec:experiments} presents experiments, and Section~\ref{sec:conclusion} concludes.

%%%%%%%%%%%%%%%%%%%%%%%%%%%%%%%%%%%%%%%%%%%%%%%%%%%%%%%%%%%
\section{Preliminaries \& Problem Formulation}
\label{sec:preliminaries}
%%%%%%%%%%%%%%%%%%%%%%%%%%%%%%%%%%%%%%%%%%%%%%%%%%%%%%%%%%%%

We introduce the background material on cellular sheaves of inner-product spaces over graphs and the nonlinear sheaf Laplacian.

\subsection{Euclidean Sheaves}

Suppose $G = (V,E)$ is a directed graph with finite vertex-set $V$, finite edge-set $E$, and maps $h,t: E \rightrightarrows V$ sending an edge $e$ to its head $h(e)$ or tail $t(e)$.
We define a partially ordered set $\mathsf{P}(G) \coloneqq (V \sqcup E, \face)$ on the disjoint union generated by the relation: $v \face e$ whenever $e \in h^{-1}(v) \cup t^{-1}(v)$.
This partial order encodes incidence relations between vertices and edges of the graph and will be utilized to index spaces and linear maps in the following definition.

\begin{definition}
    A \emph{Euclidean sheaf} over $G = (V,E)$ is an assignment
    %\footnote{For experts, a \emph{functor} $\F: \mathsf{P}(G) \to \mathsf{Hilb}$ where $\mathsf{Hilb}$ is the category of finite-dimensional real Hilbert spaces and bounded linear maps.}
    of the following to $G$:
    \begin{description}
        \item[Stalks]: A real inner-product space $\F(v)$ to every $v \in V$ with $\langle x_v,x_v', \rangle_v \coloneqq x_v^\top R_v x_v'$, $R_v \succ 0$.
        \item[Edge stalks]: A real inner-product space $\F(e)$ to every $e \in E$ with $\langle y_e,y_e', \rangle_e \coloneqq y_e^\top R_e y_e'$, $R_e \succ 0$.
        \item[Restriction maps]: A pair of bounded linear maps
            \begin{equation}
                \begin{aligned}
                    \F(h(e)) \xrightarrow{\res{\F}{h(e)}{e}} \F(e), \quad \F(t(e)) \xrightarrow{\res{\F}{t(e)}{e}} \F(e)
                \end{aligned}
            \end{equation}
            for every $e \in E$.
    \end{description}
\end{definition}

\begin{interpret}
    In a MAS, $G = (V,E)$ models the communication or flow of information between agents: $V$ indexes the agents of the system and $E$ indexes communication channels between pairs of agents. $\F(v)$ is the local state space of an agent $v \in V$. $\F(e)$ is a shared communication space for the edge $e \in E$.
    The restriction map $\res{\F}{v}{e}$ models the transmission of information about the state of an agent to a channel. The positive-definite matrices $R_v$ and $R_e$ define weighted inner products on the stalk and edge-stalk spaces, allowing different coordinates to be weighted according to their relative importance in the communication channel.
\end{interpret}

\begin{definition}
    Suppose $\F$ is a Euclidean sheaf over $G$. Let
    \begin{equation}
        \begin{aligned}
        C^0(G;\F) \coloneqq \textstyle\bigoplus_{v \in V} \F(v), \quad
        C^1(G;\F) \coloneqq \textstyle\bigoplus_{e \in E} \F(e)
        \end{aligned}
    \end{equation}
    denote the spaces of 0-cochains and 1-cochains, respectively, with inner products
    \begin{equation}
        \begin{aligned}
            \langle x, x' \rangle_{C^0} \coloneqq \textstyle\sum_{v \in V} \langle x_v, x'_v \rangle_v, \quad \langle y, y' \rangle_{C^1} \coloneqq \textstyle\sum_{e \in E} \langle y_e, y'_e \rangle_e
        \end{aligned}
    \end{equation}
   % \na{should this not be $y_e$ instead of $y_v$? since $y$ in $C^1(G;F)$ has edge-indexed components?}~\hr{good catch!}
    The \emph{coboundary map} of $\F$ is a bounded linear map $\delta_{\F}: C^0(G;\F) \to C^1(G;\F)$ defined by
    \begin{equation}
        (\delta_{\F} x)_e \coloneqq \res{\F}{h(e)}{e} x_{h(e)} - \res{\F}{t(e)}{e} x_{t(e)}.
    \end{equation}
    The kernel $\ker \delta_{\F} \subseteq C^0(G;\F)$ is the space of \emph{global sections} $H^0(G;\F)$. The \emph{first sheaf cohomology} is $H^1(G;\F) := C^1(G;\F) / \mathrm{im}\,\delta_{\F}$.
\end{definition}

\begin{notation}
    In coordinates, $\delta_\F$ (written $\delta$ when $\F$ is clear from context) is represented by a $d_1 \times d_0$ matrix $B$, and its Hilbert-space adjoint satisfies $\delta^\ast = M_1^{-1} B^\top M_2$, where $M_1 \succ 0$ and $M_2 \succ 0$ are the block-diagonal Gram matrices of $C^0(G;\F)$ and $C^1(G;\F)$, respectively.
\end{notation}
\begin{remark}
    $H^1(G;\F)$ is a quotient vector space whose elements are equivalence classes $[y] = y + \mathrm{im}\,\delta_\F$.
\end{remark}

\begin{interpret}
    In a MAS, $C^0(G;\F)$ is the (global) state space of the system. $\delta_\F$ simultaneously measures the disagreements between all agents in communication channels. $H^0(G;\F)$ collects global states $x \in C^0(G;\F)$ that are in agreement, i.e.~$\res{\F}{h(e)}{e}(x_{h(e)}) = \res{\F}{t(e)}{e}(x_{t(e)})$. $H^1(G;\F)$ measures obstructions to agreement.
\end{interpret}

A Hodge Decomposition \cite{eckmannHarmonischeFunktionenUnd1944} applies to Euclidean sheaves over graphs, implying the following lemma.

\begin{lemma}[Hodge decomposition, \cite{hansenSheafLaplacianSpectral2019}]
    \label{lem:hodge}
    $C^1(G;\F) = \mathrm{im}\,\delta \oplus \ker\,\delta^*$, so $\ker\,\delta^* \cong H^1(G;\F)$ as inner-product spaces.
\end{lemma}

%%%%%%%%%%%%%%%%%%%%%%%%%%%%%%%%%%%%%%%%%%%%%%%%%%%%%%%%%%%%
\subsection{Dynamical Systems on Sheaves}
\label{sec:sheaf-laplacian}
%%%%%%%%%%%%%%%%%%%%%%%%%%%%%%%%%%%%%%%%%%%%%%%%%%%%%%%%%%%%

We now turn to dynamical systems governed by a Laplace operator defined for a given sheaf over $G$.

\begin{definition}
    Suppose $\Phi: C^1(G;\F) \to C^1(G;\F)$ is bounded and continuous. The \emph{nonlinear sheaf Laplacian} is the operator $L^\Phi_\F: C^0(G;\F) \to C^0(G;\F)$ defined by $L^\Phi_\F \coloneqq \delta^\ast \circ \Phi \circ \delta$.
\end{definition}

The linear sheaf Laplacian $L_\F := \delta^* \circ \delta$ (corresponding to $\Phi = \mathrm{id}$) was introduced and studied in \cite{hansenSheafLaplacianSpectral2019}. Solutions of $\dot{x} = -L_\F x$ converge exponentially to the orthogonal projection of $x(0)$ onto $\ker L_\F = H^0(G;\F)$ \cite[Proposition 8.1]{hansenSheafLaplacianSpectral2019}.
The nonlinear case was first studied in \cite{hansenOpinionDynamicsDiscourse2021}, which establishes convergence to $H^0(G;\F)$ for the special case of odd, monotone $\Phi$  \cite[Proposition 10.1]{hansenOpinionDynamicsDiscourse2021}. We consider a more general assumption about $\Phi$.

\begin{assumption} \label{assume:Phi}
   There exists $U: C^1(G;\F) \to \R$ such that $\nabla U = \Phi$. Moreover, $U: C^1(G;\F) \to \R$ can be written as a sum $U(y) = \textstyle\sum_{e \in E} U_e(y_e)$.
\end{assumption}

Under \Cref{assume:Phi}, we have the following convergence result for the nonlinear sheaf diffusion equation.

\begin{theorem}[\hspace{-0.5pt}{\cite[Theorem 2]{hanksDistributedMultiAgentCoordination2025}}]
    \label{thm:hanks_convergence}
    Let $\F$ be a Euclidean sheaf on $G$ and suppose each edge potential $U_e: \F(e) \to \R$ is strongly convex with unique minimizer $b_e \in \F(e)$. Then for any diffusivity $\alpha > 0$, trajectories of
    \begin{equation}
        \dot{x} = -\alpha\, L_{\F}^{\nabla U}(x)
        \label{eq:hanks_dynamics}
    \end{equation}
    converge to the orthogonal projection of $x(0)$ onto
    \begin{equation}
        \delta_{\F}^+ b + H^0(G;\F) = \ker L_{\F}^{\nabla U},
    \end{equation}
    where $\delta_{\F}^+ := (\delta_{\F}^* \delta_{\F})^{-1} \delta_{\F}^*$ denotes the Moore--Penrose pseudoinverse of $\delta_{\F}$ (restricted to $(\ker\delta_{\F})^\perp$).
\end{theorem}

\subsection{Problem Formulation}
We now formulate our system identification problem for recovering the edge potentials from trajectories of sheaf dynamics.
The \emph{sheaf diffusion equation} is the ODE
\begin{equation}
    \dot{x} = -L_{\F}^{\Phi}(x) - \Psi(x) \label{eq:dynamics}
\end{equation}
where $\Phi$ (edge field) and $\Psi$ (node field) satisfy \Cref{assume:Phi} and \Cref{assume:Psi}.
% \hr{does it make sense to call these ``node`` and ``edge'' fields?} \na{yess} \hr{removed the $\alpha$ from the \Cref{eq:dynamics}}

% \begin{remark}
%     By \Cref{thm:hanks_convergence}, if $\Psi = 0$, solutions of~\eqref{eq:dynamics} converge to the invariant set $\ker L^{\nabla U}_\F = \delta_{\F}^+ b + H^0(G;\F)$.
%     Once the system has equilibrated, all trajectories lie in this set and $C^1_\mathrm{obs} = \delta(C^0_\mathrm{obs})$ may be confined to a low-dimensional subset of $C^1(G;\F)$, leaving the basis functions $\{\psi_m\}$ poorly excited. Identifiable recovery therefore requires either short-horizon data collected before equilibration, or deliberate perturbations away from the equilibrium manifold.
% \end{remark}

\begin{assumption} \label{assume:Psi}
   There exists $W: C^0(G;\F) \to \R$ such that $\nabla W = \Psi$, and $W$ decomposes as $W(x) = \textstyle\sum_{v \in V} W_v(x_v)$.
\end{assumption}

The following example collects noteworthy edge potential functions, including the laws studied in our experiments in Section~\ref{sec:parametric}.

\begin{example}[Edge potentials]\label{ex:edge_potentials}
    The following classes arise in coordination and opinion dynamics \cite{hansenOpinionDynamicsDiscourse2021,hanksDistributedMultiAgentCoordination2025}:
    \begin{description}
        \item[Quadratic]: $U_e(y_e) = \tfrac{1}{2}\|y_e\|^2_{\mathcal{F}(e)}$, giving $\Phi = \mathrm{id}$ and recovering the linear sheaf Laplacian $L_\mathcal{F}$.
        \item[Shifted quadratic (formation control)]: $U_e(y_e) = \tfrac{1}{2}\|y_e - \hat{y}_e\|^2_{\mathcal{F}(e)}$ for a target 1-cochain $\hat{y} \in C^1(G;\mathcal{F})$. The force $\nabla U_e(y_e) = y_e - \hat{y}_e$ drives the system toward $\delta x = \hat{y}$; when $\hat{y} \in \mathrm{im}\,\delta$ this corresponds to a prescribed formation offset $\hat{p}_{ij}$ between neighbors $i$ and $j$ \cite{hanksDistributedMultiAgentCoordination2025}.
        \item[Bounded confidence (opinion dynamics)]: $U_e(y_e) = \psi_e(\|y_e\|^2_{\mathcal{F}(e)})$ for a smooth $\psi_e:[0,\infty)\to\mathbb{R}$ with $\psi_e'(z) > 0$ for $z < D_e$ and $\psi_e'(z) = 0$ for $z \geq D_e$. Agents communicate only when disagreement is below threshold $D_e$, reproducing Hegselmann--Krause bounded-confidence dynamics on the sheaf \cite{hansenOpinionDynamicsDiscourse2021}.
        \item[Antagonistic]: $U_e(y_e) = -\|y_e\|^2_{\mathcal{F}(e)}$ for edges in a set $E_- \subseteq E$ of adversarial links. The resulting signed Laplacian $L^S_\mathcal{F} = \delta^* S\delta$ (with $S = -I$ on $E_-$, $S = I$ on $E_+$) may be indefinite, producing unstable dynamics when $E_-$ is a cutset \cite{hansenOpinionDynamicsDiscourse2021}.
    \end{description}
\end{example}

With this family of examples in view, we now formulate the central inverse problem: recovering the edge potential from trajectory observations.

\begin{problem}[System Identification] \label{problem:main}
    Given trajectory data $\{x^{(i)}(t)\}_{i=1}^M$ from solutions of \eqref{eq:dynamics} under Assumptions~\ref{assume:Phi}--\ref{assume:Psi}, recover the edge potential $U: C^1(G;\F) \to \R$, or characterize the obstructions to doing so.
\end{problem}

%%%%%%%%%%%%%%%%%%%%%%%%%%%%%%%%%%%%%%%%%%%%%%%%%%%%%%%%%%%%
\section{Non-Parametric Potential Reconstruction}
\label{sec:nonparametric}
%%%%%%%%%%%%%%%%%%%%%%%%%%%%%%%%%%%%%%%%%%%%%%%%%%%%%%%%%%%%
The SI problem reduces to recovering the edge force law $\Phi = \nabla U: C^1(G;\mathcal{F}) \to C^1(G;\mathcal{F})$ from trajectory observations. Throughout this section, we assume $\Psi$ satisfies Assumption~\ref{assume:Psi} with each node potential $W_v$ given as prior knowledge, so that the residual
\begin{equation}
    r(x) := -\dot{x} - \Psi(x) = \delta^* \Phi(\delta x)
    \label{eq:edge_residual_def}
\end{equation}
is directly computable from trajectory data.

% \na{We seem to have dropped $\alpha$ in the residual definitions. from \eqref{eq:dynamics}: section c (proble formulation), $-\dot{x}-\Psi(x)=\alpha\delta^*\Phi(\delta x)$, so we can either define $r(x):=(-\dot{x}-\Psi(x))/\alpha$ or state explicitly that $\alpha=1$ / absorbed into $\Phi$ (i think this would be the simplest fix). this also affects the least-squares and integral-fitting equations.} \hr{Yup, removed the $\alpha$ entirely}
% This is the observability condition under which system identification is feasible.

Suppose trajectories $x^{(i)}:[0,T_i]\to C^0(G;\mathcal{F})$, $i=1,\ldots,M$, are sampled from the dynamical system~\eqref{eq:dynamics} with unknown but fixed $\Phi$ and $\Psi$. Define the observed region in cochain space:
\begin{equation}
    \begin{gathered}
        C^0_\mathrm{obs} := \textstyle\bigcup_{i=1}^M x^{(i)}([0,T_i]) \subseteq C^0(G;\mathcal{F}), \\[0.5em]
        C^1_\mathrm{obs} := \delta(C^0_\mathrm{obs}) \subseteq C^1(G;\mathcal{F}).
    \end{gathered}
    \label{eq:observed_cochain_regions}
\end{equation}

Since $\Phi$ depends on $y \in C^1(G;\mathcal{F})$ alone (Assumption~\ref{assume:Phi}), the residual~\eqref{eq:edge_residual_def} factors through $\delta x$. We may therefore define the \emph{observable force map}
\begin{equation}
    f_\mathrm{obs}(y) := \delta^* \Phi(y), \qquad y \in C^1_\mathrm{obs}.
    \label{eq:observable_force}
\end{equation}
Trajectories reveal only $f_\mathrm{obs}$, not $\Phi$ itself. The SI problem reduces to inverting the map $\Phi \mapsto \delta^* \Phi$ on $C^1_\mathrm{obs}$.

\begin{lemma}[Identifiability obstruction]
    \label{lem:obstruction}
    Two edge force laws $\Phi, \widetilde{\Phi}: C^1(G;\mathcal{F}) \to C^1(G;\mathcal{F})$ satisfy $\delta^* \Phi(y) = \delta^* \widetilde{\Phi}(y)$ for all $y \in C^1_\mathrm{obs}$ if and only if $(\Phi - \widetilde{\Phi})(y) \in \ker \delta^*$ for all $y \in C^1_\mathrm{obs}$. By Lemma~\ref{lem:hodge}, $\ker \delta^* \cong H^1(G;\mathcal{F})$, so the dimension of the space of undetectable perturbations equals $\dim H^1(G;\mathcal{F})$.\footnote{For proofs and additional experimental details, see
    the Appendix.%
    % our full technical report at \url{https://hansriess.com/paper/sheaf-si/}.%
    }
\end{lemma}

The following theorem characterizes when recovery is possible without parameterization.

\begin{theorem}[Unrestricted Recovery Boundary]
    \label{thm:unrestricted_recovery}
    Suppose the sheaf $\mathcal{F}$ is fixed and known to the observer, and $\Psi$ satisfies Assumption~\ref{assume:Psi} with each $W_v$ given as prior knowledge. Let
    \begin{equation}
        \mathcal{M}_U := \bigl\{\nabla U : U(y) = \textstyle\sum_{e \in E} U_e(y_e),\ U_e \in C^1(\mathcal{F}(e))\bigr\}
    \end{equation}
    be the set of conservative, edge-separable force maps. Then the following are equivalent:
    \begin{enumerate}
        \renewcommand{\labelenumi}{\textnormal{(\roman{enumi})}}
        \item For every $\Phi \in \mathcal{M}_U$, no other $\widetilde{\Phi} \in \mathcal{M}_U$ satisfies $\delta^* \Phi(y) = \delta^* \widetilde{\Phi}(y)$ for all $y \in C^1_\mathrm{obs}$.
        \item $H^1(G;\mathcal{F}) = 0$.
    \end{enumerate}
\end{theorem}

When the obstruction vanishes, the recovered force law also determines each edge potential uniquely.

\begin{corollary}
    When the conditions of Theorem~\ref{thm:unrestricted_recovery} hold, the recovered force $\Phi = \nabla U$ uniquely determines each $U_e$ on every path-connected component of $\pi_e(C^1_\mathrm{obs})$, where $\pi_e$ is the projection onto $\mathcal{F}(e)$, up to an additive constant.
    % \hr{Maybe write this as $\pi_e(C^1_{\mathrm{obs}})$ \na{done, see if that is what you meant.}}
\end{corollary}

\section{Parametric Potential Reconstruction}
\label{sec:parametric}
%%%%%%%%%%%%%%%%%%%%%%%%%%%%%%%%%%%%%%%%%%%%%%%%%%%%%%%%%%%%

When $H^1(G;\mathcal{F}) \neq 0$, Lemma~\ref{lem:obstruction} shows that $\ker\delta^* \neq 0$, so for any $\Phi \in \mathcal{M}_U$ the perturbation $\Phi + z$ with $z \in \ker\delta^*$, $z \neq 0$, is indistinguishable from $\Phi$ via trajectory data. Recovery from $\mathcal{M}_U$ is impossible; identifiability can be restored by restricting to a finite-dimensional parameterized class.

For ease of exposition, we assume all edge stalks are isomorphic to a common space $W \cong \mathbb{R}^d$, which holds when all communication channels employ the same comparison space---that is, when all edge stalks $\mathcal{F}(e)$ are copies of the same space. The edge-heterogeneous case follows by applying the construction per isomorphism class.

Fix a finite collection of smooth basis potentials $\psi_1, \ldots, \psi_p$. With unknown coefficients $\theta = (\theta_m) \in \mathbb{R}^p$, define
\begin{equation}
    U_\theta(y) := \textstyle\sum_{e \in E} \textstyle\sum_{m=1}^{p} \theta_m\, \psi_m(y_e).
    \label{eq:parametrization}
\end{equation}
The $\theta_m$ are shared across all edges because we assume homogeneous stalks and a common interaction law: every agent pair uses the same potential family, differing only via the restriction maps encoded in $\delta$. The corresponding force $\Phi_\theta(y) = (\sum_m \theta_m \nabla \psi_m(y_e))_{e \in E}$ is linear in $\theta$.

\begin{theorem}
    \label{thm:partitioned_identifiability}
    Suppose we observe $N$ distinct edge states $y^{(1)}, \ldots, y^{(N)} \in C^1_\mathrm{obs}$ with corresponding observed node signals $f^{(i)}_\mathrm{obs} := f_\mathrm{obs}(y^{(i)})$.
    \begin{enumerate}
        \renewcommand{\labelenumi}{\textnormal{(\alph{enumi})}}
        \item \emph{Linear parameterization:} Define the design matrix $A: \mathbb{R}^p \to C^0(G;\mathcal{F})^N$ by
              \begin{equation}
                  (A\theta)_i := \delta^* \nabla U_\theta(y^{(i)}), \quad i = 1, \ldots, N.
                  \label{eq:design_matrix}
              \end{equation}
              The true parameter $\theta_*$ is uniquely recoverable from the data $\{f^{(i)}_\mathrm{obs}\}$ if and only if $\ker A = \{0\}$, equivalently, the Gram matrix $\Gamma := A^* A$ is positive definite.

        \item \emph{Nonlinear parameterization:} For a nonlinear parameterization on open $\Theta \subseteq \mathbb{R}^p$, let $\mathcal{O}_N(\theta) := \bigl(\delta^* \nabla U_\theta(y^{(i)})\bigr)_{i=1}^N \in C^0(G;\mathcal{F})^N$.
              The parameter $\theta_*$ is locally identifiable if the Jacobian $D\mathcal{O}_N(\theta_*)$ is injective, equivalently, the information matrix
              \begin{equation}
                  \mathcal{I}_N := \textstyle\sum_{i=1}^N S_i^* S_i, \quad S_i := D_\theta[\delta^* \nabla U_\theta(y^{(i)})]_{\theta = \theta_*},
                  \label{eq:information_matrix}
              \end{equation}
              is positive definite.
    \end{enumerate}
\end{theorem}

\begin{remark}
    The edge states $y^{(i)} = \delta x^{(i)}$ entering Theorem~\ref{thm:partitioned_identifiability} are computed from observed node states via the coboundary map $\delta$; no direct measurement of edge quantities is required.
\end{remark}

\begin{remark}
    Theorem~\ref{thm:partitioned_identifiability} implies that identifiability depends jointly on: (i) the choice of basis functions $\{\psi_m\}$, and (ii) the diversity of edge states $\{y^{(i)}\}$ visited by the data---richer coverage of $C^1_\mathrm{obs}$ leads to a better-conditioned $\mathcal{I}_N$, while data confined to a low-dimensional region may leave $\mathcal{I}_N$ singular even when the basis is well chosen. For a fixed basis, one can optimize $\mathcal{I}_N$ via experimental design: trajectories should visit diverse regions of $C^1_\mathrm{obs}$, and basis functions should be chosen so that their images under $\delta^*$ span a large subspace of $C^0(G;\mathcal{F})$.
\end{remark}

\begin{figure*}[!t]
    \centering
    \includegraphics[width=\textwidth,keepaspectratio]{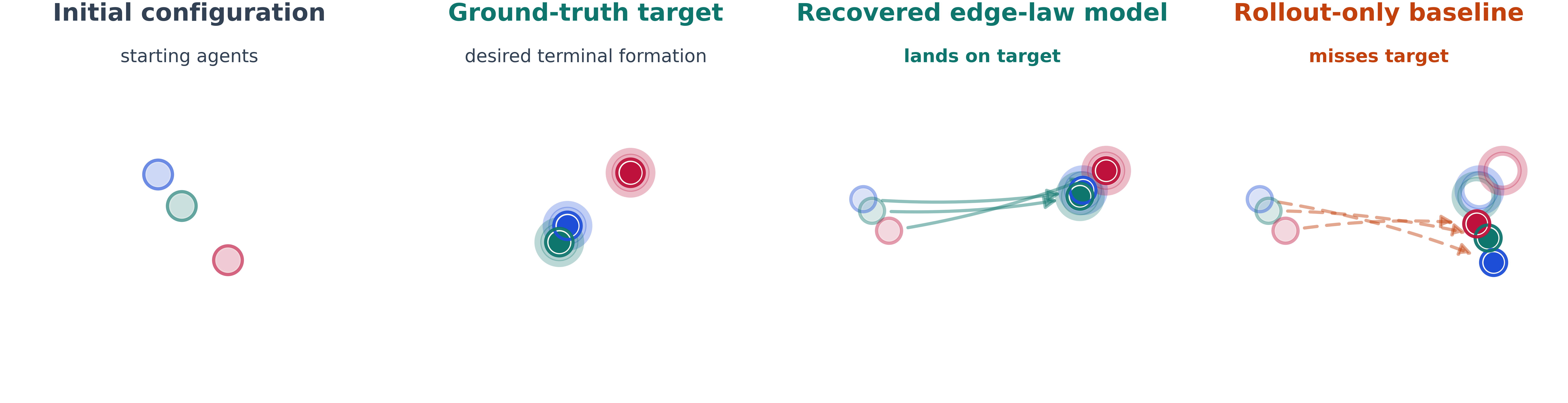}
    \caption{Experiment 1 on Sheaf B. The recovered edge law reaches the target formation; the rollout-only baseline does not.}
    \label{fig:exp1_formation_transfer}
\end{figure*}

\subsection{Estimation via Least Squares}
\label{subsec:estimation}

% \hr{Emperical risk minimization (ERM) formulation of the problem can get incorporated here}
Given the parameterization~\eqref{eq:parametrization}, we observe discrete trajectory samples $\{(x_k, \dot{x}_k)\}_{k=1}^N$ and compute residuals $r_k := -\dot{x}_k - \Psi(x_k)$, which satisfy $r_k = \delta^* \nabla U_\theta(\delta x_k)$ by~\eqref{eq:edge_residual_def}. We estimate $\theta$ via
\begin{equation}
    \hat\theta \in \arg\textstyle\min_{\theta}~(1/N) ~\textstyle\sum_{k=1}^N \bigl\|r_k - \delta^* \nabla U_\theta(\delta x_k)\bigr\|^2_{C^0} + \lambda \mathcal{R}(\theta),
    \label{eq:estimation_objective}
\end{equation}
where $\mathcal{R}$ is a regularizer and $\lambda \geq 0$. Parameterizing $U_\theta$ directly enforces conservativity by construction. When derivatives are noisy, one may instead minimize the integrated residual
\begin{equation}
    x(t_{k+1}) - x(t_k) \approx -\textstyle\int_{t_k}^{t_{k+1}} \bigl[\delta^* \nabla U_\theta(\delta x(s)) + \Psi(x(s))\bigr] \dd{s},
    \label{eq:integral_fitting}
\end{equation}
which avoids explicit differentiation and improves numerical stability.

\section{Experiments}
\label{sec:experiments}

Each experiment isolates a different obstruction to edge-law recovery---cohomological ambiguity, insufficient excitation, and finite-basis rank deficiency---and tests whether accurate node rollout implies accurate recovery of the local law. Throughout, $\Psi=0$, so $r(x)=-\dot{x}$ is directly observable.

\subsection{Experiment 1: Cohomological Ambiguity in Formation Transfer}

We study sheaf dynamics on a directed cycle under a shifted quadratic edge potential (Example~\ref{ex:edge_potentials}), whose equilibrium encodes a target agent formation. We construct a competing edge law that differs from the true law by an element of $\ker\delta^*$---a perturbation invisible to node-level observations---and compare its effect on two sheaves: Sheaf~A, for which $H^1(G;\mathcal{F})\neq 0$, and Sheaf~B, for which $H^1(G;\mathcal{F})=0$. On Sheaf~A, the perturbed law produces identical node trajectories to the true law, so a rollout criterion accepts the wrong local rule and the system fails to reach the target formation; on Sheaf~B, the same perturbation alters the rollout and is detectable. Figure~\ref{fig:exp_nonparametric_threshold_summary} (left) illustrates both cases. Consequently, when $H^1(G;\mathcal{F})\neq 0$, node-trajectory agreement does not certify edge-law recovery: distinct local laws can produce observationally equivalent dynamics.

\begin{figure}[t]
\vspace{-1em}
\centering
\includegraphics[width=\columnwidth,keepaspectratio]{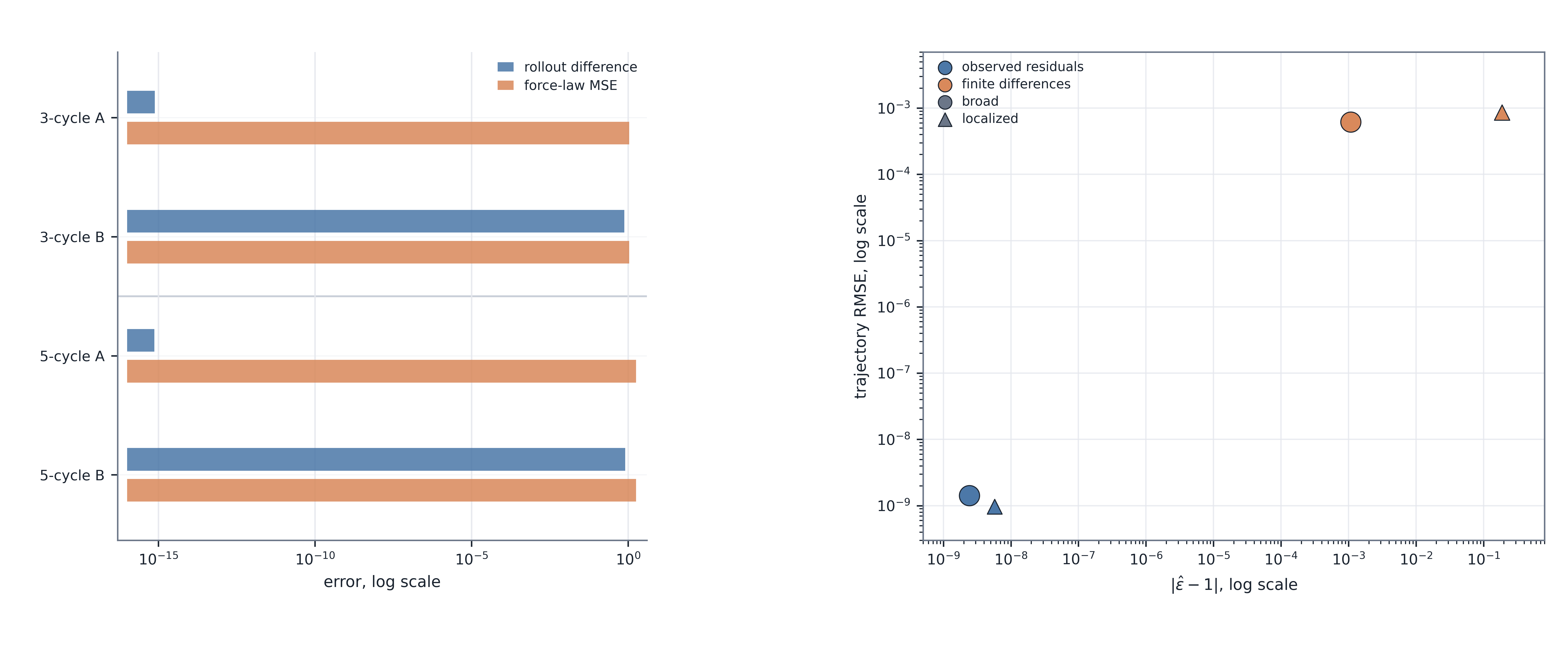}
\vspace{-2em}
\caption{Nonparametric ambiguity and threshold recovery.}
\label{fig:exp_nonparametric_threshold_summary}
\end{figure}

\subsection{Experiment 2: Parametric Recovery under Bounded Confidence}

On a sheaf with $H^1(G;\mathcal{F})=0$, we recover the scalar confidence threshold $\epsilon$ of the bounded-confidence edge potential (Example~\ref{ex:edge_potentials}), specialized to the smooth polynomial approximation
\[
    U_e(y_e;\epsilon)=\psi_\epsilon(\|y_e\|),
    \quad
\psi_\epsilon(r)=
\begin{cases}
\frac12 r^2-\frac{1}{2\epsilon^2}r^4+\frac{1}{6\epsilon^4}r^6,
& r\le \epsilon,\\[0.6em]
\frac{\epsilon^2}{6},
& r>\epsilon,
\end{cases}
\]
under which each edge contributes a restoring force only when disagreement falls below $\epsilon$. We vary data coverage (broad initial conditions spanning below, near, and above the threshold vs.\ localized conditions confined near the cutoff) and residual quality (directly observed vs.\ estimated by finite differences under node noise). Broad coverage recovers the threshold accurately under both residual regimes, while localized coverage degrades substantially when residuals are estimated from noisy data; in all cases, rollout error alone does not signal the failure. Figure~\ref{fig:exp_nonparametric_threshold_summary} (right) illustrates these outcomes. The key lesson is that sufficient coverage of $C^1_\mathrm{obs}$---in particular, excitation of edge states near the threshold---is a necessary condition for recovery independent of any structural obstruction.

\subsection{Experiment 3: Finite-Basis Identifiability}

We return to a sheaf with $H^1(G;\mathcal{F})\neq 0$, where nonparametric recovery is impossible, and restrict the edge law to the finite-dimensional monomial basis class of Section~\ref{sec:parametric}. We test three conditions against the criterion of Theorem~\ref{thm:partitioned_identifiability}: the correct basis with broad data coverage; the correct basis augmented with a harmonic edge mode (a constant force that vanishes under $\delta^*$ and is therefore unidentifiable from trajectories); and the correct basis with limited initial-condition coverage. The correct basis with broad coverage yields a positive-definite information matrix $\Gamma$ and accurate parameter recovery. The augmented basis renders $\Gamma$ singular---the coefficient of the harmonic mode is not identifiable despite small rollout error---and limited coverage makes $\Gamma$ nearly singular with similarly degraded recovery, again without a correspondingly large rollout error. Figure~\ref{fig:exp_basis_force_summary} confirms this picture. Identifiability in the parametric setting is thus governed by the rank of $\Gamma$, jointly determined by basis choice and data coverage, not by rollout accuracy.

\begin{figure}[t]
\vspace{-2em}
\centering
\includegraphics[width=\columnwidth,keepaspectratio]{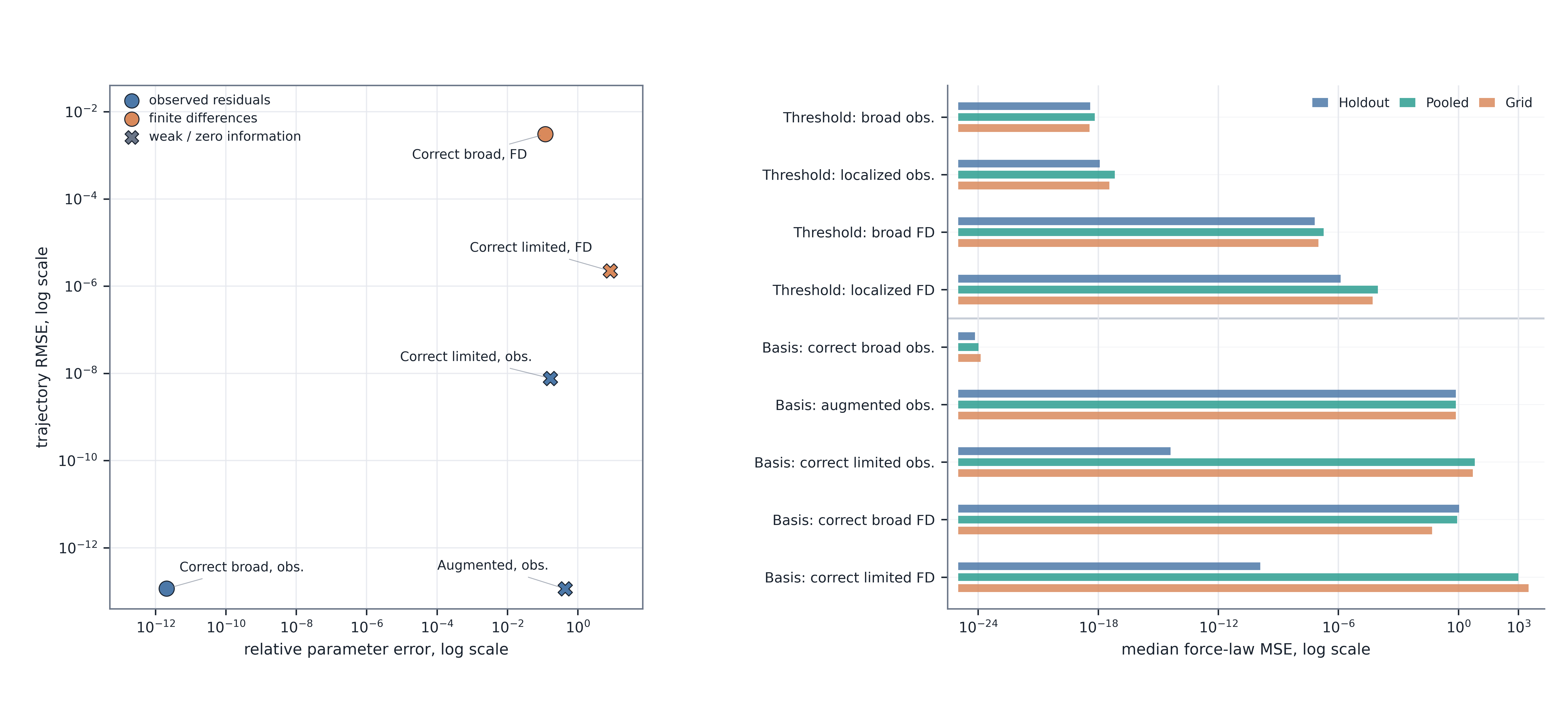}
\caption{Finite-basis recovery and force-law error.}
\label{fig:exp_basis_force_summary}
\vspace{-1em}
\end{figure}

\section{Conclusion}
\label{sec:conclusion}
We studied when local interaction laws can be recovered from trajectory data in sheaf-based multi-agent systems. Our theory shows that recovery is determined by the sheaf cohomology of the system: when $H^1(G;\mathcal{F})=0$, the local law is uniquely recoverable from trajectory data; when $H^1(G;\mathcal{F})\neq 0$, recovery requires both a restricted parametric class and sufficiently diverse data. In either regime, accurate node rollout does not certify recovery of the underlying local law---prediction and system identification are distinct objectives, and a learned model is interpretable only when the identifiability conditions are met. A natural next step is to develop estimation and experiment-design methods that remain reliable under noise, partial observation, and more general heterogeneous dynamics. A further limitation of the present work is that the node potential $\Psi$ is assumed known and subtracted from the data as preprocessing; extending the identifiability theory to the joint recovery of both $\Phi$ and $\Psi$ from trajectories is an important open problem.

%%%%%%%%% BIBLIOGRAPHY %%%%%%%%%%%%%%%%%%%%%%%%%%%%%%%%%%%%%%
\newpage
\bibliographystyle{IEEEtran}
\bibliography{references}
%%%%%%%%%%%%%%%%%%%%%%%%%%%%%%%%%%%%%%%%%%%%%%%%%%%%%%%%%%%%%

%%%%%%%%% APPENDIX %%%%%%%%%%%%%%%%%%%%%%%%%%%%%%%%%%%%%%%%%%
\appendix
\subsection{Proofs}
%%%%%%%%%%%%%%%%%%%%%%%%%%%%%%%%%%%%%%%%%%%%%%%%%%%%%%%%%%%%%

This appendix proves Lemma~\ref{lem:obstruction}, Theorem~\ref{thm:unrestricted_recovery},
its corollary, and Theorem~\ref{thm:partitioned_identifiability}. Throughout, gradients
and adjoints are taken with respect to the inner products on the stalk and cochain spaces.

\begin{proof}[Proof of Lemma~\ref{lem:obstruction}]
Fix $y\in C^1_\mathrm{obs}$. Then
\begin{equation}
    \begin{aligned}
    &\delta^*\Phi(y)=\delta^*\widetilde{\Phi}(y)\\
    &\iff \delta^*(\Phi-\widetilde{\Phi})(y)=0\\
    &\iff (\Phi-\widetilde{\Phi})(y)\in\ker\delta^*.
    \end{aligned}
\end{equation}
Since this holds for every $y\in C^1_\mathrm{obs}$, the first claim follows.

By Lemma~\ref{lem:hodge}, $\ker\delta^*\cong H^1(G;\F)$ as inner-product spaces.
Hence, for each fixed observed edge state $y$, the vector space of undetectable
perturbation values at $y$ is exactly $\ker\delta^*$, and therefore has dimension
$\dim H^1(G;\F)$.
\end{proof}

\begin{proof}[Proof of Theorem~\ref{thm:unrestricted_recovery}]
We prove \textnormal{(i)} \(\Leftrightarrow\) \textnormal{(ii)}.

\emph{\textnormal{(ii)} $\Rightarrow$ \textnormal{(i)}.}
Assume $H^1(G;\F)=0$. By Lemma~\ref{lem:hodge}, this is equivalent to
$\ker\delta^*=\{0\}$. If $\Phi,\widetilde{\Phi}\in\mathcal{M}_U$ satisfy
\[
\delta^*\Phi(y)=\delta^*\widetilde{\Phi}(y)
\qquad\text{for all }y\in C^1_\mathrm{obs},
\]
then Lemma~\ref{lem:obstruction} gives
\[
(\Phi-\widetilde{\Phi})(y)\in\ker\delta^*=\{0\}
\qquad\text{for all }y\in C^1_\mathrm{obs}.
\]
Hence $\Phi(y)=\widetilde{\Phi}(y)$ for all $y\in C^1_\mathrm{obs}$, proving
\textnormal{(i)}.

\emph{\textnormal{(i)} $\Rightarrow$ \textnormal{(ii)}.}
We prove the contrapositive. Suppose $H^1(G;\F)\neq 0$. By Lemma~\ref{lem:hodge},
choose $0\neq z=(z_e)_{e\in E}\in\ker\delta^*$. Fix any $\Phi=\nabla U\in\mathcal{M}_U$,
where
\[
U(y)=\sum_{e\in E}U_e(y_e),
\qquad
U_e\in C^1(\F(e)).
\]
For each edge $e$, define
\[
\begin{aligned}
\ell_e(\xi)&:=\langle z_e,\xi\rangle_e,\\
\widetilde{U}_e&:=U_e+\ell_e,\\
\widetilde{U}(y)&:=\sum_{e\in E}\widetilde{U}_e(y_e).
\end{aligned}
\]
Then each $\widetilde{U}_e$ is $C^1$, so $\widetilde{\Phi}:=\nabla\widetilde{U}$
belongs to $\mathcal{M}_U$. Since $\nabla \ell_e=z_e$, we have
\begin{equation}
    \widetilde{\Phi}(y)
    =
    \nabla\widetilde{U}(y)
    =
    \nabla U(y)+z
    =
    \Phi(y)+z
\end{equation}
for all $y\in C^1(G;\F)$.
Thus $\widetilde{\Phi}\neq\Phi$. But for every $y\in C^1_\mathrm{obs}$,
\begin{equation}
    \delta^*\widetilde{\Phi}(y)
    =\delta^*(\Phi(y)+z)
    =\delta^*\Phi(y)+\delta^*z
    =\delta^*\Phi(y),
\end{equation}
because $z\in\ker\delta^*$. Hence \textnormal{(i)} fails. This proves the
contrapositive.
\end{proof}

\begin{proof}[Proof of the Corollary]
Under Theorem~\ref{thm:unrestricted_recovery}, the force law $\Phi=\nabla U$ is
uniquely determined on $C^1_\mathrm{obs}$. Fix $e\in E$ and let
\[
S_e:=\pi_e(C^1_\mathrm{obs})\subseteq\F(e).
\]
Because $U$ is edge-separable,
\[
\Phi_e(y)=\nabla U_e(y_e)
\qquad\text{for all }y\in C^1(G;\F).
\]
Hence if $y,\bar y\in C^1_\mathrm{obs}$ satisfy $\pi_e(y)=\pi_e(\bar y)$, then
\[
\Phi_e(y)=\nabla U_e(y_e)=\nabla U_e(\bar y_e)=\Phi_e(\bar y).
\]
Therefore
\[
g_e:S_e\to\F(e),
\qquad
g_e(\xi):=\Phi_e(y),
\]
for any $y\in C^1_\mathrm{obs}$ with $\pi_e(y)=\xi$,
is well defined, and $g_e=\nabla U_e$ on $S_e$.

Let $Q$ be a piecewise-$C^1$ path-connected component of $S_e$, and fix $\xi_0\in Q$.
For any $\xi\in Q$, choose a piecewise-$C^1$ path $\gamma:[0,1]\to Q$ with
$\gamma(0)=\xi_0$ and $\gamma(1)=\xi$. Along each $C^1$ segment of $\gamma$, the chain rule gives
\[
\frac{\dd}{\dd t}U_e(\gamma(t))
=
\langle \nabla U_e(\gamma(t)),\dot\gamma(t)\rangle_e
=
\langle g_e(\gamma(t)),\dot\gamma(t)\rangle_e.
\]
Integrating along the segments yields
\[
U_e(\xi)-U_e(\xi_0)
=
\textstyle\int_0^1 \langle g_e(\gamma(t)),\dot\gamma(t)\rangle_e\,\dd t.
\]
Thus $g_e$ determines $U_e$ on $Q$ up to the additive constant $U_e(\xi_0)$.

If $\widehat U_e\in C^1(\F(e))$ also satisfies $\nabla \widehat U_e=g_e$ on $Q$, then
$h:=U_e-\widehat U_e$ satisfies $\nabla h=0$ on $Q$. For any piecewise-$C^1$ path
$\gamma$ in $Q$ from $\xi_0$ to $\xi$,
\[
h(\xi)-h(\xi_0)
=
\textstyle\int_0^1 \langle \nabla h(\gamma(t)),\dot\gamma(t)\rangle_e\,\dd t
=
0.
\]
Hence $h$ is constant on $Q$, so $\widehat U_e$ differs from $U_e$ by an additive
constant on $Q$.
\end{proof}

\begin{proof}[Proof of Theorem~\ref{thm:partitioned_identifiability}]
We treat \textnormal{(a)} and \textnormal{(b)} separately.

\emph{\textnormal{(a)} Linear parameterization.}
For each $i\in\{1,\dots,N\}$ and $m\in\{1,\dots,p\}$, set
\[
\zeta_m^{(i)}:=\bigl(\nabla\psi_m(y_e^{(i)})\bigr)_{e\in E}\in C^1(G;\F).
\]
Then
\[
\begin{aligned}
\nabla U_\theta(y^{(i)})
&=\sum_{m=1}^p \theta_m\,\zeta_m^{(i)},\\
\delta^*\nabla U_\theta(y^{(i)})
&=\sum_{m=1}^p \theta_m\,\delta^*\zeta_m^{(i)}.
\end{aligned}
\]
Hence the map $A:\mathbb{R}^p\to C^0(G;\F)^N$ defined by
\[
A\theta:=
\bigl(
\delta^*\nabla U_\theta(y^{(1)}),\dots,
\delta^*\nabla U_\theta(y^{(N)})
\bigr),
\]
is linear. In the noiseless setting,
\[
A\theta_*=(f_\mathrm{obs}^{(1)},\dots,f_\mathrm{obs}^{(N)}).
\]
Thus $\theta_*$ is uniquely determined by the data if and only if
$A\theta=A\theta_*$ implies $\theta=\theta_*$, i.e., if and only if
$\ker A=\{0\}$.

Equip $C^0(G;\F)^N$ with the product inner product and $\mathbb{R}^p$ with its
standard Euclidean inner product. Then, for every $v\in\mathbb{R}^p$,
\[
\langle v,\Gamma v\rangle
=
\langle v,A^*Av\rangle
=
\langle Av,Av\rangle
=
\|Av\|^2.
\]
Therefore $\Gamma=A^*A$ is positive definite if and only if
$Av\neq 0$ for every $v\neq 0$, equivalently, if and only if $\ker A=\{0\}$.

\emph{\textnormal{(b)} Nonlinear parameterization.}
Let
\[
J:=D\mathcal{O}_N(\theta_*):\mathbb{R}^p\to C^0(G;\F)^N.
\]
Assume $J$ is injective. Then $J$ is an isomorphism from $\mathbb{R}^p$ onto the
$p$-dimensional subspace $\operatorname{im}J\subseteq C^0(G;\F)^N$, so there exists a
linear map
\[
P:C^0(G;\F)^N\to\mathbb{R}^p
\]
such that $PJ=I_{\mathbb{R}^p}$. Define $F:=P\circ\mathcal{O}_N$. Then $F$ is $C^1$ and
\[
DF(\theta_*)=P\,D\mathcal{O}_N(\theta_*)=PJ=I_{\mathbb{R}^p}.
\]
By the inverse function theorem, $F$ is locally one-to-one near $\theta_*$. Hence, if
$\theta_1,\theta_2$ are sufficiently close to $\theta_*$ and
$\mathcal{O}_N(\theta_1)=\mathcal{O}_N(\theta_2)$, then
\[
F(\theta_1)=P\mathcal{O}_N(\theta_1)=P\mathcal{O}_N(\theta_2)=F(\theta_2),
\]
so $\theta_1=\theta_2$. Thus $\mathcal{O}_N$ is locally injective at $\theta_*$, and
$\theta_*$ is locally identifiable.

Finally, by definition of $S_i$,
\[
Jv=(S_1v,\dots,S_Nv)
\qquad\text{for all }v\in\mathbb{R}^p.
\]
Therefore
\[
\begin{aligned}
\|Jv\|^2
&=
\sum_{i=1}^N \|S_iv\|_{C^0}^2\\
&=
\sum_{i=1}^N \langle v,S_i^*S_iv\rangle\\
&=
\left\langle v,\left(\sum_{i=1}^N S_i^*S_i\right)v\right\rangle\\
&=
\langle v,\mathcal{I}_N v\rangle.
\end{aligned}
\]
Hence $\mathcal{I}_N$ is positive definite if and only if $J$ is injective, which is
the claimed equivalence.
\end{proof}

\subsection{Additional experimental details}

All node and edge stalks are two-dimensional; trajectories are integrated with step size $0.01$. Rollout errors compare node trajectories; force-law errors compare the recovered interaction in edge space. Results for Experiments~2 and~3 are reported as mean $\pm$ standard deviation over eight independent seeds, except the deterministic augmented-basis row in Table~\ref{tab:exp_basis}.

\subsection*{Experiment 1: Setup Details}

The experiment uses three agents on a directed $3$-cycle. The shifted-quadratic edge potential is
\[
    U_e(y_e; b_e) = \tfrac{1}{2}\|y_e - b_e\|^2, \qquad \nabla U_e(y_e; b_e) = y_e - b_e.
\]
The true system is simulated for $4$ seconds; its terminal state defines the target formation. Agents are then reset and coordinated with two candidate laws: the recovered-edge-law model, and a rollout-only alternative $\widetilde{\Phi}(y) = y + \beta c$ with $\beta = 0.6$ and $c = (1,0)$ on each edge. Sheaf~A uses identity tail maps ($\dim H^1(G;\mathcal{F}) = 2$); Sheaf~B uses rotated tail maps ($H^1(G;\mathcal{F}) = 0$). On Sheaf~A, the perturbation $\beta c \in \ker\delta^*$ is invisible to trajectory observations; on Sheaf~B, the same perturbation alters the rollout and the terminal formation. Numerical results are reported in Table~\ref{tab:exp_nonparametric}.

\subsection*{Experiment 2: Setup Details}

We use the rotated $3$-cycle sheaf ($H^1(G;\mathcal{F})=0$) and set the true threshold $\epsilon_\star = 1$. Broad sampling uses initial-condition scales $(0.4, 0.8, 1.2)$, placing edge states below, near, and above the threshold; localized sampling uses only a thin annulus near the cutoff. Residuals are either observed directly or estimated by finite differences after adding Gaussian node noise with standard deviation $5\times10^{-3}$. Table~\ref{tab:exp_bc} reports recovery error, rollout RMSE, and information number $\mathcal{I}_N$ across conditions. Force-law MSE on pooled and grid evaluation sets is reported in Table~\ref{tab:exp_force_checks}. On the $5$-cycle under the same sheaf, broad sampling with observed residuals gives $|\hat{\epsilon}-\epsilon_\star| = 2.21\times10^{-9}$, while localized sampling with finite-difference residuals gives $2.54\times10^{-2}$.

\subsection*{Experiment 3: Setup Details}

We use the identity $3$-cycle sheaf ($\dim H^1(G;\mathcal{F}) = 2$) with the monomial edge law
\begin{fitdisplaymath}
    \Phi_\theta(y_e) = y_e\bigl(\theta_1 + \theta_2\|y_e\|^2 + \theta_3\|y_e\|^4\bigr), \qquad \theta_\star = (1,\, 0.25,\, 0.03).
\end{fitdisplaymath}
For finite-difference residuals, Gaussian node noise with standard deviation $10^{-4}$ is added. The augmented basis appends a constant force $c \in \ker\delta^*$ that is invisible to $\delta^*$ and therefore produces a singular $\Gamma$. Table~\ref{tab:exp_basis} reports relative parameter error, rollout RMSE, and $\lambda_{\min}(\Gamma)$ across conditions. Table~\ref{tab:exp_force_checks} shows that models with limited coverage can match holdout trajectories while failing on pooled observed edge states and a reference grid. On the identity $5$-cycle, the relative parameter errors for the correct/broad/Obs., augmented/Obs., and correct/limited/Obs.\ conditions are $1.15\times10^{-12}$, $4.36\times10^{-1}$, and $1.18\times10^{-1}$, respectively.

\FloatBarrier

\begin{table}[!b]
\caption{Nonparametric ambiguity on cycle sheaves. Rollout difference measures node trajectories; force MSE measures edge laws.}
\label{tab:exp_nonparametric}
\centering
\footnotesize
\setlength{\tabcolsep}{4pt}
\renewcommand{\arraystretch}{1.05}
\begin{adjustbox}{max width=\linewidth}
\begin{tabular}{@{}lccc@{}}
\hline
Sheaf & $\dim H^1$ & Max rollout diff. & Force MSE \\
\hline
3-cycle, Sheaf A & 2 & $7.73\times 10^{-16}$ & $1.08$ \\
3-cycle, Sheaf B & 0 & $7.60\times 10^{-1}$  & $1.08$ \\
5-cycle, Sheaf A & 2 & $7.47\times 10^{-16}$ & $1.80$ \\
5-cycle, Sheaf B & 0 & $8.07\times 10^{-1}$  & $1.80$ \\
\hline
\end{tabular}
\end{adjustbox}
\end{table}

\begin{table}[!b]
\caption{Recovery of the bounded-confidence threshold on the rotated $3$-cycle sheaf. Obs. = observed residuals; FD = finite differences. Entries are mean $\pm$ standard deviation over eight seeds.}
\label{tab:exp_bc}
\centering
\scriptsize
\setlength{\tabcolsep}{3.5pt}
\renewcommand{\arraystretch}{1.05}
\begin{adjustbox}{max width=\linewidth}
\begin{tabular}{@{}lccc@{}}
\hline
Setting & $|\hat{\epsilon}-1|$ & Rollout RMSE & $\mathcal I_N$ \\
\hline
Broad / Obs. & $(2.43 \pm 1.05)\times 10^{-9}$ & $(1.43 \pm 1.22)\times 10^{-9}$ & $(2.03 \pm 0.30)\times 10^{3}$ \\
Localized / Obs. & $(5.74 \pm 3.07)\times 10^{-9}$ & $(9.67 \pm 5.13)\times 10^{-10}$ & $(2.08 \pm 0.09)\times 10^{2}$ \\
Broad / FD & $(1.08 \pm 0.54)\times 10^{-3}$ & $(6.16 \pm 5.81)\times 10^{-4}$ & $(1.99 \pm 0.29)\times 10^{3}$ \\
Localized / FD & $(1.88 \pm 3.12)\times 10^{-1}$ & $(8.56 \pm 1.98)\times 10^{-4}$ & $(2.86 \pm 0.09)\times 10^{2}$ \\
\hline
\end{tabular}
\end{adjustbox}
\end{table}

\begin{table}[!b]
\caption{Recovery in the finite basis class on the identity $3$-cycle sheaf. Obs. = observed residuals; FD = finite differences. Entries are mean $\pm$ standard deviation over eight seeds, except the deterministic augmented-basis row.}
\label{tab:exp_basis}
\centering
\scriptsize
\setlength{\tabcolsep}{3.5pt}
\renewcommand{\arraystretch}{1.05}
\begin{adjustbox}{max width=\linewidth}
\begin{tabular}{@{}lccc@{}}
\hline
Setting & Rel. param. err. & Rollout RMSE & $\lambda_{\min}(\Gamma)$ \\
\hline
Correct / Broad / Obs. & $(2.07 \pm 0.32)\times 10^{-12}$ & $(1.16 \pm 0.13)\times 10^{-13}$ & $(4.12 \pm 0.79)\times 10^{3}$ \\
Augmented / Obs. & $4.36\times 10^{-1}$ & $1.15\times 10^{-13}$ & $0$ \\
Correct / Limited / Obs. & $(1.67 \pm 0.20)\times 10^{-1}$ & $(7.68 \pm 1.03)\times 10^{-9}$ & $(9.29 \pm 6.67)\times 10^{-17}$ \\
Correct / Broad / FD & $(1.20 \pm 0.71)\times 10^{-1}$ & $(3.06 \pm 1.56)\times 10^{-3}$ & $(2.33 \pm 0.30)\times 10^{3}$ \\
Correct / Limited / FD & $8.41 \pm 5.79$ & $(2.23 \pm 1.74)\times 10^{-6}$ & $(6.84 \pm 5.26)\times 10^{-17}$ \\
\hline
\end{tabular}
\end{adjustbox}
\end{table}

\begin{table}[!b]
\caption{Median force-law MSE on three evaluation sets for the $3$-cycle. The pooled set and reference grid expose errors hidden by holdout trajectory error.}
\label{tab:exp_force_checks}
\centering
\scriptsize
\setlength{\tabcolsep}{3.2pt}
\renewcommand{\arraystretch}{1.05}
\begin{adjustbox}{max width=\linewidth}
\begin{tabular}{@{}lccc@{}}
\hline
Setting & Holdout & Pooled & Grid \\
\hline
\multicolumn{4}{@{}l}{\emph{Bounded-confidence threshold recovery}} \\
Broad / Obs. & $3.99\times 10^{-19}$ & $6.78\times 10^{-19}$ & $3.67\times 10^{-19}$ \\
Localized / Obs. & $1.21\times 10^{-18}$ & $6.69\times 10^{-18}$ & $3.62\times 10^{-18}$ \\
Broad / FD & $6.55\times 10^{-8}$ & $1.85\times 10^{-7}$ & $1.00\times 10^{-7}$ \\
Localized / FD & $1.30\times 10^{-6}$ & $9.41\times 10^{-5}$ & $5.25\times 10^{-5}$ \\
\hline
\multicolumn{4}{@{}l}{\emph{Finite basis class}} \\
Correct / Broad / Obs. & $6.97\times 10^{-25}$ & $1.03\times 10^{-24}$ & $1.34\times 10^{-24}$ \\
Augmented / Obs. & $7.50\times 10^{-1}$ & $7.50\times 10^{-1}$ & $7.50\times 10^{-1}$ \\
Correct / Limited / Obs. & $4.15\times 10^{-15}$ & $6.53\times 10^{0}$ & $5.26\times 10^{0}$ \\
Correct / Broad / FD & $1.08\times 10^{0}$ & $8.76\times 10^{-1}$ & $4.91\times 10^{-2}$ \\
Correct / Limited / FD & $1.25\times 10^{-10}$ & $9.78\times 10^{2}$ & $3.19\times 10^{3}$ \\
\hline
\end{tabular}
\end{adjustbox}
\end{table}

\end{document}